\documentclass[10pt,twocolumn,twoside]{IEEEtran}
\IEEEoverridecommandlockouts 

\usepackage{graphicx}
\usepackage{color}
\usepackage{subfigure}
\usepackage{algorithm,algpseudocode}
\usepackage{amsmath, amssymb, amsthm, mathtools, accents}

\usepackage{amsmath,amssymb}
\usepackage[noadjust]{cite}
\usepackage{tikz}
	\usetikzlibrary{shapes,arrows}
	\tikzstyle{frame} = [draw, -latex]
	\tikzstyle{line} = [draw]
	\tikzstyle{line2} = [draw, dashdotted]
	\tikzstyle{line3} = [draw, dashed]
	\tikzstyle{line3UD} = [draw, dashed]
	\tikzstyle{place} = [circle, draw=black, fill=white, thick, inner sep=2pt, minimum size=1mm]
	\tikzstyle{place2} = [circle, draw=black, fill=black, thick, inner sep=2pt, minimum size=1mm]
	\tikzstyle{placeRed} = [circle, draw=red, fill=red, thick, inner sep=2pt, minimum size=1mm]
	\tikzstyle{vertex} = [circle, draw=black, fill=black, thick, inner sep=2pt, minimum size=1mm]
\usepackage{tkz-euclide}

\usetikzlibrary{shapes,arrows} 

\tikzstyle{decision} = [diamond, draw, fill=blue!20,
    text width=4.5em, text badly centered, node distance=3cm, inner sep=0pt]
\tikzstyle{block1} = [rectangle, draw, text width=8em, text centered, minimum height=4em]
\tikzstyle{block2} = [rectangle, draw, text width=3em, text centered, minimum height=4em]
\tikzstyle{block3} = [rectangle, draw, text width=11em, text centered, minimum height=12em, dashed,black]
\tikzstyle{block4} = [rectangle, draw, text width=11em, text centered, minimum height=18em, dashed,black]
\tikzstyle{block5} = [rectangle, draw, text width=11em, text centered, minimum height=32em, dashed,black]
\tikzstyle{block6} = [rectangle, draw, text width=11em, text centered, minimum height=18.5em, dashed,black]
\tikzstyle{block7} = [rectangle, draw, text width=11em, text centered, minimum height=11.8em, dashed,black]
\tikzstyle{line01} = [draw, -latex']
\tikzstyle{line02} = [draw, latex'-latex']

\newcommand{\bbm}{\begin{bmatrix}}
\newcommand{\ebm}{\end{bmatrix}}

\usepackage{amsmath,bm}

\newtheorem{corollary}{\bfseries Corollary}
\newtheorem{definition}{\bfseries Definition }
\newtheorem{lemma}{\bfseries Lemma}

\newtheorem{remark}{\bfseries Remark}
\newtheorem{theorem}{\bfseries Theorem}
\newtheorem{assumption}{\bfseries Assumption}

\makeatletter
\def\BState{\State\hskip-\ALG@thistlm}
\makeatother
\algdef{SE}[DOWHILE]{Do}{doWhile}{\algorithmicdo}[1]{\algorithmicwhile\ #1}%

\definecolor{MG}{rgb}{0,0.45,0.08}

\graphicspath{{img/}}

\title{Topological Controllability of Undirected Networks of Diffusively-Coupled Agents}

\author{Hyo-Sung Ahn$^{1,2}$,  Kevin L. Moore$^{2}$, Seong-Ho Kwon$^{1}$, Quoc Van Tran$^{1}$, Byeong-Yeon Kim$^{3}$, and Kwang-Kyo Oh$^{4}$

\thanks{${}^{1}$School of Mechanical Engineering, Gwangju Institute of Science and Technology (GIST), Gwangju, Korea. E-mails: {\tt\small hyosung@gist.ac.kr; seongho@gist.ac.kr; tranvanquoc@gist.ac.kr}}
\thanks{${}^{2}$Department of Electrical Engineering, Colorado School of Mines, Golden, CO, USA. E-mails: {\tt\small kmoore@mines.edu; hahn@mines.edu}}
\thanks{${}^{3}$Korea Atomic Energy Research Institute, Daejeon, Korea.  E-mail: {\tt\small tktrktna12@gmail.com}}
\thanks{${}^{4}$Department of Electrical Engineering, Sunchon National University, Sunchon, Jeollanam-do, Korea.  E-mail: {\tt\small kwangkyo.oh@gmail.com}}
}

\begin{document}
\maketitle
\thispagestyle{empty}
\pagestyle{empty}

\begin{abstract} This paper presents conditions for establishing topological controllability in undirected networks of diffusively-coupled agents. Specifically, controllability is considered based on the signs of the edges (negative, positive or zero). Our approach differs from well-known structural controllability conditions for linear systems or consensus networks, where controllability conditions are based on edge connectivity (i.e., zero or non-zero edges).  Our results first provide a process for  merging controllable graphs into a larger controllable graph. Then, based on this process, we provide a graph decomposition process for evaluating the topological controllability of a given network.
\end{abstract}
\begin{IEEEkeywords}
Diffusive networks, Topological controllability, Structural controllability, Merging process, Decomposition process
\end{IEEEkeywords}

\section{Introduction}
This paper studies the controllability of a class of network systems using only knowledge of the sign-connectivity between nodes, without relying upon knowledge of the magnitude of the connections. By \textit{sign-connectivity}, we mean that the knowledge of signs of edges can be known; but the magnitude of the edges are not known. Since the magnitude of the edges weights are not used, it is not an algebraic approach, but rather we call it a topological approach. We call such a topological analysis of the controllability of a network \textit{topological controllability}.

Topological controllability has also been studied under the name of \textit{structural controllability} \cite{ ChingTaiLin_TAC_1974}, although they do not consider the signs of the edges. In traditional controllability of a network or in linear systems theory \cite{Antsaklis:2007:LSP:1543534}, both the topology and coupling strengths (i.e., magnitude of the connections) between nodes are taken into account. Thus, in traditional approaches, it is a topological and algebraic solution. However, in certain networks, it may be hard to know the coupling strengths between nodes or there may be uncertainties in the measure or identification of the coupling strengths. Consequently, if the controllability of a network can be evaluated only using the topology or structure, it will be beneficial in some applications, including control of digital and electric circuits \cite{Goldstein_TCS_1979,Feng_Lu_ICMLC_2005}, power electronics \cite{Summers_etal_TCNS_2015}, large scale networks \cite{Wang_automatica_2016}, complex networks \cite{Wang_etal_SR_2017,Shen_etal_CCDC_2018},  and brain networks \cite{Gu_etal_nature_comm_2015}.

From a review of the literature, we could find many interesting analyses and concepts related to structural controllability or topological controllability. The pioneering analysis was conducted in \cite{ ChingTaiLin_TAC_1974}, which introduced the concepts of stem, dilations, bud, origin, cactus, and accessibility. These concepts were used for merging basic controllable graphs into a bigger graph. Since  \cite{ ChingTaiLin_TAC_1974}, there has been a lot of research on the controllability of network systems. In \cite{YangYu_nature_2011}, the concept of maximum matching was introduced to find matched nodes from inputs. The matched nodes are elements of paths. Then, it was argued that unmatched nodes need to be controlled directly by control inputs.  In \cite{Ruths_science_2014}, a minimum control structure, i.e., the minimum number of independent control inputs, was further examined on the basis of number of source nodes, and external/internal dilation nodes. Under Laplacian dynamics, input symmetry was characterized for making a network uncontrollable in \cite{Mesbahi_Egerstedt_book}. We also note that there have been many studies on structural controllability for specific types of dynamic systems, including switching networks \cite{Hou_TCS1_2016}, high-order dynamic systems \cite{Partovi_etal_ICRAM_2010}, random networks \cite{Faradonbeh_TCNS_2017}, and descriptor systems \cite{Yang_etal_ICCA_2007}. 

However, in most of the literature, only the connectivity between nodes are considered. That is, in most of literature, a value of an edge is given as zero or non-zero. But, in many network systems the signs of edge (i.e., positive or negative) are quite critical for evaluating the convergence or stabilization of the overall network. For example, in diffusive coupling networks, the positive edge means a cooperative coupling, while a negative edge means a negative coupling. In some systems (e.g., social networks) there can exist edges with different signs, as some agents are cooperative and others are antagonistic. For such systems  it is critical to know the signs of the edge values. Motivated from this observation, in this paper, we assume that edges of a network are classified as negative, positive, or zero and with only this knowledge we provide conditions for the topological controllability of the network system (for more motivations and advantage of using signs of edges, see \textit{Remark~\ref{remark_contribution}}). To deliver our ideas in a simple way, we consider the specific case of diffusively-coupled agents connected as an undirected network, although the techniques developed in this paper can be extended to directed and general networks.

In the sequel there are two main results. First, we interpret the results of \cite{Tsatsomeros_siam_1998} in the sense of a graph. We then present some conditions for merging subgraphs under the condition of topological controllability. Then, by the merging rules, we can gradually enlarge the network while keeping the topological controllability. However, this merging process does not provide a direct method for evaluating a topological controllability of a given network. Thus, as the second goal of this paper, we present some ideas to decompose a graph into subgraphs, which are path graphs. Then, starting again from the decomposed subgraphs, we gradually again add the edges to merge the subgraphs under the topological controllability condition. By this way, we can find a largest subgraph, which can be called a subgraph induced by the controllability. This allows us to develop an algorithm for testing the topological controllability of a given network. 

The paper consists of as follows. In Section~\ref{sec_pre}, some preliminaries are given and the topological controllability problems are formulated. In Section~\ref{sec_topo}, certain conditions for topological controllability are presented, and in Section~\ref{sec_graph}, two algorithms are provided to examine the topological controllability of a given network.  Examples and conclusions are presented in Section~\ref{sec_exam} and Section~\ref{sec_conc} respectively.

\section{Preliminaries and Problem formulations} \label{sec_pre}
Let an undirected network of diffusively-coupled agents $x_i$ with direct nodes inputs $u_i$ be given by:
\begin{align}
\dot{x}_i = -\sum_{j \in \mathcal{N}_i} a_{ij} (x_i - x_j) + b_i u_i
\end{align}
where $ \mathcal{N}_i$ is the set of neighboring nodes of $i$, and $a_{ij}$ are diffusive couplings and $b_i$ are input couplings. We define the network $T= [L, B]$ concisely as the Laplacian dynamics:
\begin{align}
\dot{x} = L x + B u \label{eq_laplacian}
\end{align}
where $x=(x_1,\ldots, x_n)^T$, $u=(u_1, \ldots, u_m)^T$, $L \in \Bbb{R}^{n \times n}$ is a Laplacian matrix with possible negative edges, and $B  \in \Bbb{R}^{n \times m} $ is  the input coupling matrix. The Laplacian matrix $L$ is a matrix defined by the interactions of  $n$ state nodes and the matrix $B$ defines input couplings from  $m$ input nodes to state nodes.  So, there are $n+m$ nodes in the network. The interactions among state nodes are undirected (thus L is a symmetric row- and column-stochastic matrix) while the interactions from the input nodes to state nodes are directed. It is also assumed that each input node is connected to only one state node by one-to-one mapping (injective).

\begin{definition}
Controllability: An undirected network $T= [L, B]$ of diffusively-coupled agents with directed input nodes given by \ref{eq_laplacian} is said to be controllable if there exists an input vector $u(t)$ such that $x(t)\to x^*$ for any desired vector $x^*$.\end{definition}

\begin{definition}
Topological controllability: A controllable undirected network $T= [L, B]$ of diffusively-coupled agents with directed input nodes given by \ref{eq_laplacian} is said to be topologically controllable if all other undirected networks $\bar T= [ \bar L, \bar B]$ whose edges have the same signs (positive, negative, or zero) as $T= [L, B]$ are also controllable.
\end{definition}

To characterize the topological controllability of an undirected network of diffusively-coupled agents, we borrow the analysis given in \cite{Tsatsomeros_siam_1998}. Thus, this paper is a kind of
an interpretation of the analaysis of \cite{Tsatsomeros_siam_1998}. Let the network can be re-defined  as a \textit{graph}, denoted
\begin{align}
 \mathcal{G}(T) = (\mathcal{V}, \mathcal{E})
\end{align}
where $T= [L, B]$, the set of vertices $\mathcal{V}$ is the set of indices of nodes as $\mathcal{V}=\{ \underbrace{ 1, \ldots, n}_{\text{state nodes} = \mathcal{V}^S }, \underbrace{n+1, \ldots, n+m}_{ \text{input nodes} = \mathcal{V}^I}\}$, and the set of edges $\mathcal{E}$ is determined from the interaction characteristics between nodes. Fig.~\ref{network_graph_concept} depicts a network and a graph.
It is necessary to distinguish the concepts of \textit{network} and \textit{graph}. The network is a relationship of physical interactions among nodes, while the graph is a representation of the network as a set of vertices and edges.  To illustrate, consider a network depicted in Fig.~\ref{network_graph_concept}(a). 
With some edge weightings, for example, let  the Laplacian matrix corresponding to the network in Fig.~\ref{network_graph_concept}(a) be given as:
\begin{align}
L = \left[\begin{array}{ccccc}
               -2 & 2 & 1 & 0 & -1  \\
              2 & -3 & 1 & 1 & -1  \\
              1 & 1 & -3 & 1 & 0  \\
              0 & 1 & 1 & -5 & 3  \\
              -1 & -1 & 0 & 3 & -1  \end{array} \right]
\end{align}
and the input coupling matrix $B$ be given as:
\begin{align}
B= \left[\begin{array}{ccc}
0 & 0 & 0 \\
 0 & 0 & 0 \\
 1 & 0 & 0 \\
 0 & 1 & 0 \\
 0 & 0 & 1 \end{array} \right]
\end{align}
Then, the interaction  characteristics of a graph, which is a representation of a network, are decided by the matrices ${L}$ and $B$. That is, given $T=[t_{ij}]=[L, B]$, if $t_{ij} \neq 0$, then there exists an edge $(i,j)$, which is the directed edge from $i$ to $j$.
For undirected edges (i.e., when $i, j \in \{1, \ldots, n\}$), if there exists $(i,j)$ in $\mathcal{E}$, then there also exists $(j, i)$. If $t_{ii} \neq 0,~i \in \{1, \ldots, n\}$, then there exists a self-loop at node $i$. We assume there is no edge between the input nodes. In the graph, there are edges from $\mathcal{V}^S$ to $\mathcal{V}^I$ as $(i,j) \in \mathcal{E}$ where $i \in \mathcal{V}^S$ and $j \in \mathcal{V}^I$. The edge $(i,j)$  from a state node to an input node in the graph implies that the node $i$ is influenced by $j$. For a node $i$, if there exists an edge $(i,j)$, then $j$ is a neighboring node (the set of neighboring nodes of node $i$ is denoted as $\mathcal{N}_i$) in the graph $\mathcal{G}$, i.e., $j \in \mathcal{N}_i$. Fig.~\ref{network_graph_concept}(b) depicts a graph, which is a representation of the network in Fig.~\ref{network_graph_concept}(a). The edge directions in the graph and the network are reversed. It is shown that $\mathcal{V}^S =\{ 1, 2, 3, 4, 5  \}$ and $\mathcal{V}^I =\{6, 7, 8  \}$, and the edges from $\mathcal{V}^S$ to $\mathcal{V}^I$ are $(3,6), (4,7), (5,8)$. For a set $\alpha$, which is a subset of $\mathcal{V}$ (i.e., $\alpha \subseteq \mathcal{V}$), the set of neighboring nodes of the set $\alpha$ is  defined as $\mathcal{N}(\alpha)= \cup \mathcal{N}_i,~\forall i \in \alpha$.

The graph $\mathcal{G}$ can be decomposed as $\mathcal{G} = \mathcal{G}^S \cup  \mathcal{G}^I$, where $\mathcal{G}^S$ is the induced subgraph by $\mathcal{V}^S$, and $\mathcal{G}^I$ is the interaction graph between the set of vertices $\mathcal{V}^S$ and set of vertices $\mathcal{V}^I$. Thus, $\mathcal{G}^S = (\mathcal{V}^S, \mathcal{E}^S)$ and $\mathcal{G}^I = (\mathcal{V}, \overrightarrow{\mathcal{E}}^I)$, where $\overrightarrow{\mathcal{E}}^I$ is the set of directed edges. Note that $\mathcal{E}=\mathcal{E}^S \cup \overrightarrow{\mathcal{E}}^I$.
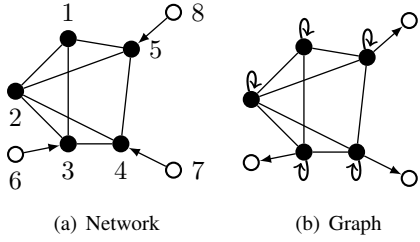
\begin{figure}
\centering
\subfigure[Network]{
\begin{tikzpicture}[scale=0.7]
\node[place, black] (node1) at (-1,1) [label=above:$1$] {};
\node[place, black] (node2) at (-2,0) [label=below:$2$] {};
\node[place, black] (node3) at (-1,-1) [label=below:$3$] {};
\node[place, black] (node4) at (0,-1) [label=below:$4$] {};
\node[place, black] (node5) at (0.2,0.8) [label=right:$5$] {};

\node[place, circle] (node6) at (-2.0,-1.2) [label=below:$6$] {};
\node[place, circle] (node7) at (1,-1.5) [label=right:$7$] {};
\node[place, circle] (node8) at (1.0,1.5) [label=right:$8$] {};

\draw (node1) [line width=0.5pt] -- node [left] {} (node2);
\draw (node1) [line width=0.5pt] -- node [right] {} (node3);
\draw (node1) [line width=0.5pt] -- node [right] {} (node5);
\draw (node2) [line width=0.5pt] -- node [right] {} (node3);
\draw (node2) [line width=0.5pt] -- node [left] {} (node4);
\draw (node2) [line width=0.5pt] -- node [left] {} (node5);
\draw (node3) [line width=0.5pt] -- node [left] {} (node4);
\draw (node4) [line width=0.5pt] -- node [left] {} (node5);

\draw (node6) [-latex, line width=0.5pt] -- node [right] {} (node3);
\draw (node7) [-latex, line width=0.5pt] -- node [right] {} (node4);
\draw (node8) [-latex, line width=0.5pt] -- node [right] {} (node5);
\end{tikzpicture}
}
\subfigure[Graph]{
\begin{tikzpicture}[scale=0.7]
\node[place, black] (node1) at (-1,1) [label=above:$$] {};
\node[place, black] (node2) at (-2,0) [label=below:$$] {};
\node[place, black] (node3) at (-1,-1) [label=below:$$] {};
\node[place, black] (node4) at (0,-1) [label=below:$$] {};
\node[place, black] (node5) at (0.2,0.8) [label=right:$$] {};

\node[place, circle] (node6) at (-2.0,-1.2) [label=below:$$] {};
\node[place, circle] (node7) at (1,-1.5) [label=right:$$] {};
\node[place, circle] (node8) at (1.0,1.5) [label=right:$$] {};

\draw (node1) [line width=0.5pt] -- node [left] {} (node2);
\draw (node1) [line width=0.5pt] -- node [right] {} (node3);
\draw (node1) [line width=0.5pt] -- node [right] {} (node5);
\draw (node2) [line width=0.5pt] -- node [right] {} (node3);
\draw (node2) [line width=0.5pt] -- node [left] {} (node4);
\draw (node2) [line width=0.5pt] -- node [left] {} (node5);
\draw (node3) [line width=0.5pt] -- node [left] {} (node4);
\draw (node4) [line width=0.5pt] -- node [left] {} (node5);

\draw (node3) [-latex, line width=0.5pt] -- node [right] {} (node6);
\draw (node4) [-latex, line width=0.5pt] -- node [right] {} (node7);
\draw (node5) [-latex, line width=0.5pt] -- node [right] {} (node8);

\draw (node1) edge [loop above,thick] node {$$} (node1);
\draw (node2) edge [loop above,thick] node {$$} (node2);
\draw (node3) edge [loop below,thick] node {$$} (node3);
\draw (node4) edge [loop below,thick] node {$$} (node4);
\draw (node5) edge [loop above,thick] node {$$} (node5);

\end{tikzpicture}
}
\caption{(a) A network with five state nodes and three input nodes. (b) Graph representation of the network.}
\label{network_graph_concept}
\end{figure}
For Fig.~\ref{network_graph_concept}, the matrix $T$ is a $5 \times (5+3)$ matrix, i.e., $T \in \Bbb{R}^{5 \times (5+3)}$.

Next, we say that any matrix with the same sign as $T$ is contained in the set of sign pattern matrices $Q(T)$. So, any matrix $T' \in Q(T)$ has the same sign as $T$ in an elementwise fashion. We also say that if the row vectors of $T'$, $\forall T' \in Q(T)$ are linearly independent, then the matrix $T$ is called an $L$-matrix. From the perspective of control system design, since the matrix $B$ can be designed, we assume that the input coupling matrix $B$ is fixed, while the Laplacian matrix is a sign pattern matrix. Thus,  $Q(T)$ is defined as
\begin{align}
Q(T) := [Q({L}), B]
\end{align}
The matrix $T=[L, B]$ is called nominal graph matrix and $Q(T)$ is called a family of sign pattern matrices.
It is certain that $\text{rank}(T) =n$ if and only if the row vectors are linearly independent. The following assumptions are necessary for simplicity.
\begin{assumption} \label{assum_diagonal}
The values of off-diagonal elements of $L$ may change; but their signs do not change (i.e., sign fixed). The diagonal elements, $l_{ii} = -\sum_{j \in \mathcal{N}_i} a_{ij}$ where $a_{ij}$ are edge weights of the network, of $L$ are non-zero and also sign fixed.
\end{assumption}
This assumption means that the sign of the summation of incident edge weights does not vary, even though each edge weight does vary under the same sign.
\begin{assumption} \label{assum_accesible}
Given a nominal graph matrix  $T=[L,B]$, the Laplacian dynamics \eqref{eq_laplacian} is controllable.
\end{assumption}
\begin{assumption} \label{assum_L_matrix}
For any $T' \in Q(T)$, the row vectors of $T'$ are linearly independent.
\end{assumption}
It is clear that these assumptions are necessary conditions for ensuring controllability for all $T' \in Q(T)$. In  \cite{Tsatsomeros_siam_1998}, \textit{Assumption~\ref{assum_accesible}} is required to ensure \textit{accessibility}\footnote{Accessibility means that  for any $i \in \mathcal{V}^S$, there is a path from $i$ to $j \in \mathcal{V}^I$ in the graph $\mathcal{G}$.} of the graph $\mathcal{G}$. If there is no path connecting an input node to a state node, the state is not controllable. The \textit{Assumption~\ref{assum_L_matrix}} means that the matrix $T=[L,B]$ is an $L$-matrix. 
\textit{Assumption~\ref{assum_accesible}} and  \textit{Assumption~\ref{assum_L_matrix}} are basic requirements for ensuring the topological controllability of a graph.

\begin{remark}
To guarantee the $L$-matrixness of $T$, one idea is to design the matrix $B$. For example,
from the relationship:
\begin{align}
\text{rank}{[T]} &=\text{rank}[{L}, B] \nonumber\\
                      &= \text{rank}({L}) +  \text{rank}(B)~~\text{if}~~ \mathbf{R}({L}) \cap  \mathbf{R}(B) = \emptyset
\end{align}
where  $ \mathbf{R}(\cdot)$ is the range of the matrix $\cdot$,
if $\text{rank}({L})= n- d$,  it is required to design $B$ such that $\text{rank}(B)= d$ with the property $ \mathbf{R}({L}) \cap  \mathbf{R}(B) = \emptyset$.
\end{remark}




With the above assumptions, the following theorem for the topological controllability of a graph is given in \cite{Tsatsomeros_siam_1998} as a sufficient condition.
\begin{theorem} \cite{Tsatsomeros_siam_1998}
Let us suppose that
 \textit{Assumption~\ref{assum_diagonal}},  \textit{Assumption~\ref{assum_accesible}}, and  \textit{Assumption~\ref{assum_L_matrix}} are satisfied. Then, for all $\alpha \subseteq \mathcal{V}^S$ satisfying $\alpha \subset \mathcal{N}(\alpha)$ in $\mathcal{G}$, if there exists at least one  $j \in  \mathcal{N}(\alpha)\setminus \alpha$ and there exists exactly one $i \in \alpha$ such that $(i,j) \in \mathcal{E}$ exists, then the graph $\mathcal{G}(T)$ determined from $T=[L,B]$ is topologically controllable. \label{theorem_Tsatsomeros}
\end{theorem}




\section{Topologically Controllable Graphs} \label{sec_topo}
This section is dedicated to an elaboration of the condition of \textit{Theorem~\ref{theorem_Tsatsomeros}}. The condition of \textit{Theorem~\ref{theorem_Tsatsomeros}} can be modified from an algorithm perspective as:
\begin{corollary} \label{coro_key}
Under the same conditions as \textit{Theorem~\ref{theorem_Tsatsomeros}}, $\forall \alpha \subseteq \mathcal{V}^S$ satisfying $\alpha \subset \mathcal{N}(\alpha)$, if there exists
$i \in \alpha$ such that $\mathcal{N}_i \cap (\mathcal{N}(\alpha)\setminus \alpha) \neq \emptyset$ and $\{\mathcal{N}_i \cap (\mathcal{N}(\alpha)\setminus \alpha)\} \setminus  \{     \mathcal{N}_j \cap (\mathcal{N}(\alpha)\setminus \alpha), \forall j \in \alpha\setminus\{i\}   \}  \neq \emptyset $, then  the graph $\mathcal{G}(T)$ determined from $T=[L,B]$ is topologically controllable. \label{corollary_Tsatsomeros}
\end{corollary}
It is relatively easy to check the statement of \textit{Corollary~\ref{corollary_Tsatsomeros}}, since we examine $i \in \alpha$ rather than $j \in \mathcal{N}(\alpha)\setminus \alpha$. It means that if there exists $i\in \alpha$, which is connected to $j \in \mathcal{N}(\alpha)\setminus \alpha$ and $j$ is not connected to other nodes in  $\alpha$, the statement of \textit{Corollary~\ref{corollary_Tsatsomeros}} is satisfied. Such a node, i.e., node $j$, is called a \textit{dedicated node} to $i$. Consequently, for any $i \in \alpha$ (at least one $i$, i.e, $\exists i \in \alpha$), if there exists a dedicated node $j \in \mathcal{N}(\alpha)\setminus \alpha$, then the grouping $\alpha$ is considered to satisfy the statement. We call a graph  topologically controllable if the condition of \textit{Corollary~\ref{corollary_Tsatsomeros}} is satisfied.
\begin{remark}
A sufficient condition for satisfying the condition of \textit{Corollary~\ref{coro_key}} is that  there exists
$i \in \alpha$ such that $\mathcal{N}_i \cap (\mathcal{N}(\alpha)\setminus \alpha) \neq \emptyset$ and $\{\mathcal{N}_i \cap (\mathcal{N}(\alpha)\setminus \alpha)\} \cap  \{     \mathcal{N}_j \cap (\mathcal{N}(\alpha)\setminus \alpha), \forall j \in \alpha\setminus\{i\}   \}  = \emptyset $.
\end{remark}

\begin{lemma}
Let us suppose that any diagonal element of ${L}$ is not identically zero. Then, under the undirected interactions in $\mathcal{V}^S$ and directed interactions between $\mathcal{V}^S$ and $\mathcal{V}^I$, for any choice $\alpha$, it is true that $\alpha \subset \mathcal{N}(\alpha)$.
\end{lemma}
\begin{proof}
For any ${L}'$, ${L}' \in Q({L})$, since the diagonal elements are non-zero, all the state nodes have self-loops. Thus, each state node has at least two neighboring nodes including itself, if the underlying graph is connected. Also, when $\alpha = \mathcal{V}^S$, the neighboring set $\mathcal{N}(\alpha)$ includes all the nodes in $\mathcal{V}^S$ and at least one node in $\mathcal{V}^I$. Thus, $\alpha \subset \mathcal{N}(\alpha)$.
\end{proof}
The above lemma shows that we need to check whether each $\alpha$, for all $\alpha \subseteq \mathcal{V}^S$, would satisfy the condition of \textit{Corollary~\ref{corollary_Tsatsomeros}}. For example, let  $\mathcal{G}^S$ be a path graph, or a tree graph, with an input at terminal node. Fig.~\ref{network_ex_path}(a) shows a path graph. In this case, whatever taking $\alpha$, it is clear that $
\{\mathcal{N}_i \cap (\mathcal{N}(\alpha)\setminus \alpha)\} \cap  \{     \mathcal{N}_j \cap (\mathcal{N}(\alpha)\setminus \alpha) \} = \emptyset $ for any $i, j \in \alpha$. That is, $\exists i \in \alpha$ and $j, \forall j \in \alpha \setminus\{i\}$ such that $\mathcal{N}_i \cap (\mathcal{N}(\alpha) \setminus \alpha) \neq \emptyset$.  Consequently, a path graph is topologically controllable, which is coincident with the result in \cite{YangYu_nature_2011}.
Fig.~\ref{network_ex_path}(b) shows a tree graph. The node $3$ is devided into two paths, i.e., $3 \leftrightarrow 1$ and $3 \leftrightarrow 2$, where the symbol $\leftrightarrow$ is used to denote the connection in the undirected path.
In this tree, if we take $\alpha=\{1, 2\}$, then the nodes $1$ and $2$ share a common neighboring node $3$, and they do not have any dedicated node. Thus, in general, a tree graph with a single input node is not topologically controllable.

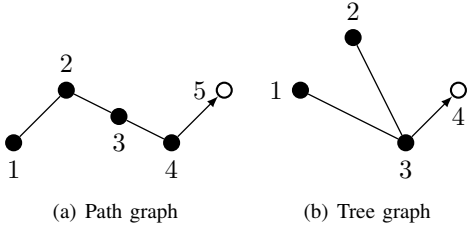
\begin{figure}
\centering
\subfigure[Path graph]{
\begin{tikzpicture}[scale=0.7]
\node[place, black] (node1) at (-2,0) [label=below:$1$] {};
\node[place, black] (node2) at (-1,1) [label=above:$2$] {};
\node[place, black] (node3) at (0,0.5) [label=below:$3$] {};
\node[place, black] (node4) at (1,0) [label=below:$4$] {};
\node[place, circle] (node5) at (2,1) [label=left:$5$] {};

\draw (node1) [line width=0.5pt] -- node [left] {} (node2);
\draw (node2) [line width=0.5pt] -- node [right] {} (node3);
\draw (node3) [line width=0.5pt] -- node [left] {} (node4);
\draw (node4) [-latex, line width=0.5pt] -- node [right] {} (node5);
\end{tikzpicture}
}
\subfigure[Tree graph]{
\begin{tikzpicture}[scale=0.7]
\node[place, black] (node1) at (-2,0) [label=left:$1$] {};
\node[place, black] (node2) at (-1,1) [label=above:$2$] {};
\node[place, black] (node3) at (0,-1) [label=below:$3$] {};
\node[place, circle] (node4) at (1,0) [label=below:$4$] {};

\draw (node1) [line width=0.5pt] -- node [right] {} (node3);
\draw (node2) [line width=0.5pt] -- node [left] {} (node3);
\draw (node3) [-latex, line width=0.5pt] -- node [right] {} (node4);
\end{tikzpicture}
}
\caption{Graphs without cycle (for a simplicity, the self-loops are  omitted in the figure).}
\label{network_ex_path}
\end{figure}

Let  $\mathcal{G}$ be a undirected cycle graph. Then, the condition is also not satisfied, without properly located input nodes. The graphs depicted in Fig.~\ref {network_ex1} include an undirected cycle. For Fig.~\ref{network_ex1}(a), when choosing $\alpha=\{1, 2\}$, the nodes $1$ and $2$ share $3$ as the common node in $\mathcal{N}(\alpha)\setminus \alpha$. So, it does not satisfy the condition. For Fig.~\ref{network_ex1}(b), we have two input nodes. When choosing $\alpha=\{1, 2\}$, the nodes $1$ and $2$ share $3$ as the common node in $\mathcal{N}(\alpha)\setminus \alpha$; but the node $1$ has a dedicated node $5$. In more detail, when choosing $\alpha =\{1,2\}$, we obtain $\mathcal{N}(\alpha) \setminus \alpha =\{3,5\}$. For $i=1$ and $j=2$, we obtain $\mathcal{N}_i =\{2,3,5\}$ and $\mathcal{N}_j=\{1,3\}$. Then, it follows that  $\mathcal{N}_i \cap (\mathcal{N}(\alpha)\setminus \alpha) \neq \emptyset$ and $\{\mathcal{N}_i \cap (\mathcal{N}(\alpha)\setminus \alpha)\} \setminus  \{     \mathcal{N}_j \cap (\mathcal{N}(\alpha)\setminus \alpha), \forall j \in \alpha\setminus\{i\}   \}  =\{5\}  \neq \emptyset$. Likewise, we can see that, for $\alpha=\{1\}$, $\alpha=\{2\}$, $\alpha=\{3\}$, $\alpha=\{2,3\}$, $\alpha=\{1,3\}$, and $\alpha=\{1,2,3\}$, there is at least one dedicated node. Thus, the graph in Fig.~\ref{network_ex1}(b) satisfies the condition. However, when a node is added between the nodes $1$ and $3$, as shown in Fig.~\ref{network_ex3},  the graph does not satisfy the condition, i.e., if we choose $\alpha=\{2, 6\}$, then the nodes $2$ and $6$ share $1$ as a commone node and $3$ also as a common node, i.e., there is no dedicated node for $2$ or for $6$.
It is remarkable that a directed cycle, with the same directions, satisfies the controllability condition since whatever choosing $\alpha$, there is a dedicated node for at least one $i \in \alpha$ (such a directed cycle is called bud in \cite{ ChingTaiLin_TAC_1974}).

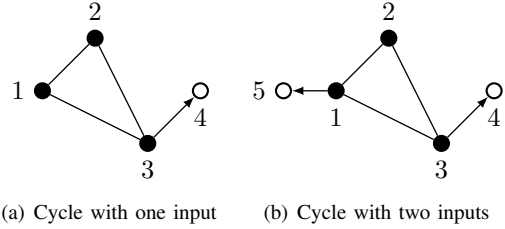
\begin{figure}
\centering
\subfigure[Cycle with one input]{
\begin{tikzpicture}[scale=0.7]
\node[place, black] (node1) at (-2,0) [label=left:$1$] {};
\node[place, black] (node2) at (-1,1) [label=above:$2$] {};
\node[place, black] (node3) at (0,-1) [label=below:$3$] {};
\node[place, circle] (node4) at (1,0) [label=below:$4$] {};


\draw (node1) [line width=0.5pt] -- node [left] {} (node2);
\draw (node1) [line width=0.5pt] -- node [right] {} (node3);
\draw (node2) [line width=0.5pt] -- node [left] {} (node3);
\draw (node3) [-latex, line width=0.5pt] -- node [right] {} (node4);
\end{tikzpicture}
}
\subfigure[Cycle with two inputs]{
\begin{tikzpicture}[scale=0.7]
\node[place, black] (node1) at (-2,0) [label=below:$1$] {};
\node[place, black] (node2) at (-1,1) [label=above:$2$] {};
\node[place, black] (node3) at (0,-1) [label=below:$3$] {};
\node[place, circle] (node4) at (1,0) [label=below:$4$] {};
\node[place, circle] (node5) at (-3,0) [label=left:$5$] {};

\draw (node1) [line width=0.5pt] -- node [left] {} (node2);
\draw (node1) [line width=0.5pt] -- node [right] {} (node3);
\draw (node2) [line width=0.5pt] -- node [left] {} (node3);
\draw (node3) [-latex, line width=0.5pt] -- node [right] {} (node4);
\draw (node1) [-latex, line width=0.5pt] -- node [right] {} (node5);

\end{tikzpicture}
}

\caption{Graphs with three state nodes (cycle) and one input node (Left), or  two input nodes (Right).}
\label{network_ex1}
\end{figure}

\begin{figure}
\centering
\begin{tikzpicture}[scale=0.7]
\node[place, black] (node1) at (-2,0) [label=below:$1$] {};
\node[place, black] (node2) at (-1,1) [label=above:$2$] {};
\node[place, black] (node3) at (0,-1) [label=below:$3$] {};
\node[place, black] (node6) at (-1,-1) [label=below:$6$] {};

\node[place, circle] (node4) at (1,0) [label=below:$4$] {};
\node[place, circle] (node5) at (-3,0) [label=left:$5$] {};

\draw (node1) [line width=0.5pt] -- node [left] {} (node2);
\draw (node1) [line width=0.5pt] -- node [right] {} (node6);
\draw (node3) [line width=0.5pt] -- node [right] {} (node6);
\draw (node2) [line width=0.5pt] -- node [left] {} (node3);
\draw (node3) [-latex, line width=0.5pt] -- node [right] {} (node4);
\draw (node1) [-latex, line width=0.5pt] -- node [right] {} (node5);
\end{tikzpicture}
\caption{A graph with four state nodes (cycle) and two input node.}
\label{network_ex3}
\end{figure}
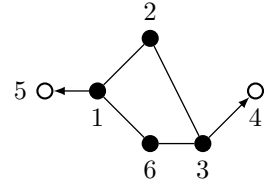

As analyzed in the above examples, it is hard to generate a general rule for the topological controllability. 
It is observed that the graph in Fig.~\ref{network_ex1}(b) is a merged graph of two paths $5 \leftarrow 1 \leftrightarrow 2$ and  $4 \leftarrow 3$ where the symbol $\leftarrow$ is used to denote a connection in directed connection in a path. Also, the graph in Fig.~\ref{network_ex3} is a merged graph of two paths $5 \leftarrow 1 \leftrightarrow 2$ and  $4 \leftarrow 3  \leftrightarrow  6$. The graph in Fig.~\ref{network_ex1}(b) is topologically controllable, while the graph in Fig.~\ref{network_ex3} is not  topologically controllable. If we can generate a graph by merging simple graphs, we may obtain some general rules.

\begin{figure}
\centering
\begin{tikzpicture}[scale=0.7]
\node[place, black] (node1) at (-2,0) [label=below:$1$] {};
\node[place, black] (node2) at (-1,1) [label=above:$2$] {};
\node[place, black] (node3) at (0,0.5) [label=below:$3$] {};
\node[place, black] (node4) at (1,0) [label=below:$4$] {};
\node[place, circle] (node5) at (2,1) [label=left:$5$] {};

\draw (node1) [line width=0.5pt] -- node [left] {} (node2);
\draw (node2) [line width=0.5pt] -- node [right] {} (node3);
\draw (node3) [line width=0.5pt] -- node [left] {} (node4);
\draw (node4) [-latex, line width=0.5pt] -- node [right] {} (node5);

\node[place, black] (node1a) at (-2,-2) [label=below:$6$] {};
\node[place, black] (node2a) at (-1,-1) [label=right:$7$] {};
\node[place, black] (node3a) at (0,-3) [label=below:$8$] {};
\node[place, circle] (node4a) at (1,-2) [label=below:$9$] {};
\node[place, circle] (node5a) at (-3,-2) [label=left:$10$] {};

\draw (node1a) [line width=0.5pt] -- node [left] {} (node2a);
\draw (node1a) [line width=0.5pt] -- node [right] {} (node3a);
\draw (node2a) [line width=0.5pt] -- node [left] {} (node3a);
\draw (node3a) [-latex, line width=0.5pt] -- node [right] {} (node4a);
\draw (node1a) [-latex, line width=0.5pt] -- node [right] {} (node5a);



\draw (node2) [line width=0.5pt, dashed] -- node [left] {} (node2a);

\end{tikzpicture}
\caption{A graph merged by two topologically controllable graphs.}
\label{network_ex_twopath}
\end{figure}
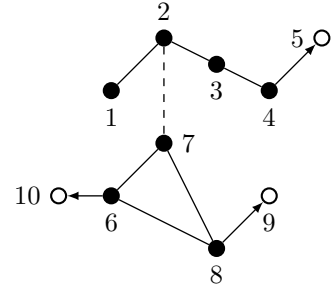

\begin{lemma}\label{lemma_two_quasi}
Let there be two disconnected topologically controllable graphs $\mathcal{G}_1$ and  $\mathcal{G}_2$. If a state node $i$ in $\mathcal{G}_1$ and a state node $j$ in $\mathcal{G}_2$ are connected by an undirected edge, then the merged graph $\mathcal{G} = \mathcal{G}_1 \cup \mathcal{G}_2 $ is topologically controllable.
\end{lemma}
\begin{proof}
When $\alpha=\{i, j\}$, $i \in  \mathcal{G}_1$ and $j \in  \mathcal{G}_2$,  it is true that  $\{\mathcal{N}_i \cap (\mathcal{N}(\alpha)\setminus \alpha)\} \cap  \{     \mathcal{N}_j \cap (\mathcal{N}(\alpha)\setminus \alpha) \}  = \emptyset      $ since they do not share a common neighbor. Also all the nodes other than $i$ in $ \mathcal{G}_1$ and all the nodes other than $j$ in $ \mathcal{G}_2$ do not have a common neighbor node (for example, as shown in Fig.~\ref{network_ex_twopath}, the nodes $2$ and $7$ do not have a common neighbor node).

Let us choose arbitrary $\alpha \subseteq  \mathcal{G}$, where $\alpha = \alpha_1 \cup \alpha_2$, and $\alpha_1 \subseteq  \mathcal{G}_1$ and $\alpha_2 \subseteq  \mathcal{G}_2$. When we choose $\alpha = \alpha_1$ or $\alpha = \alpha_2$, for any $i \in \alpha$, there is at least one dedicated node $j \in \mathcal{N}(\alpha)\setminus \alpha$ since $ \mathcal{G}_1$
and $\mathcal{G}_2$ are topologically controllable. In the case there exist $i$ and $j$ such that $i, j \in \alpha$, and $i \in \alpha_1$ and $j \in \alpha_2$, there  is still no chance of having  $\{\mathcal{N}_i \cap (\mathcal{N}(\alpha)\setminus \alpha)\} \setminus  \{     \mathcal{N}_j \cap (\mathcal{N}(\alpha)\setminus \alpha), \forall i, j  \}  = \emptyset$. 
Moreover, for all $\alpha_1 \subset \alpha$, and for all $\alpha_2  \subset \alpha$, it is certain that either in $\alpha_1$ or in $\alpha_2$,  there is a node that has a dedicated node in $\mathcal{N}(\alpha_1)\setminus \alpha_1$ or in  $\mathcal{N}(\alpha_2)\setminus \alpha_2$, respectively. Thus, the merged graph $\mathcal{G}$ is topologically controllable.
\end{proof}

With the above lemma, the following corollary is directly obtained.
\begin{corollary} \label{corollary_two_paths}
Let there be two disconnected path graphs $\mathcal{G}_1$ and  $\mathcal{G}_2$. If a state node $i$ in $\mathcal{G}_1$ and a state node $j$ in $\mathcal{G}_2$ are connected by an undirected edge, then the merged graph $\mathcal{G} = \mathcal{G}_1 \cup \mathcal{G}_2 $ is topologically controllable.
\end{corollary}
\begin{remark} \label{remark_contribution}
In structural controllability, a path with an input node (called stem) and directed cycle with the same direction with an input node (called bud) are basic controllable elements \cite{ ChingTaiLin_TAC_1974}. In maximum mathcing process \cite{YangYu_nature_2011}, the key issue is to find paths that are controllable. Similarly to the structural controllability, in topological controllability, the paths are key elements for enlarging the network. However, in our approach, i.e., topological controllability, we are not limited to the paths. The path graph is a special case for controllable graphs.  That is, although the path graphs are important for enlarging a graph (in \textit{Corollary~\ref{corollary_two_paths}} and in \textit{Algorithm~\ref{path_search}} and \textit{Algorithm~\ref{graph_merging}}), as far as the condition of \textit{Corollary~\ref{coro_key}} is satisfied, any graph can be used as a basic element for controllable graph or for enlarging the network. This superiority, in fact, can be used for merging two controllable graphs in a much general way than the cases in structural controllability, as stated in \textit{Corollary~\ref{corollary_main2}}.
\end{remark}

\begin{figure}
\centering
\begin{tikzpicture}[scale=0.7]
\node[place, black] (node1) at (-2,1) [label=above:$i_1$] {};
\node[place, black] (node2) at (-2,0) [label=below:$j_1$] {};
\node[place, black] (node3) at (-3,1.5) [label=below:$$] {};
\node[place, black] (node4) at (-2.5,0.1) [label=below:$$] {};
\node[place, black] (node5) at (-4,1.5) [label=left:$$] {};
\node[place, black] (node6) at (-3.3,-1.0) [label=left:$$] {};

\draw (node1) [line width=0.5pt] -- node [left] {} (node2);
\draw (node1) [line width=0.5pt] -- node [right] {} (node3);
\draw (node2) [line width=0.5pt] -- node [left] {} (node4);
\draw (node3) [line width=0.5pt] -- node [right] {} (node6);
\draw (node4) [line width=0.5pt] -- node [left] {} (node5);
\draw[black,dashed] (-3,0.5) circle (2) ;
\node[place, white] (node7) at (-3.0,-1.6) [label=below:$\mathcal{G}_1$] {};

\node[place, black] (node1a) at (2,0.5) [label=below:$j_2$] {};
\node[place, black] (node2a) at (2,1.8) [label=right:$i_2$] {};
\node[place, black] (node3a) at (3,1.6) [label=below:$$] {};
\node[place, black] (node4a) at (2.5,0.3) [label=below:$$] {};
\node[place, black] (node5a) at (4,1.8) [label=left:$$] {};
\node[place, black] (node6a) at (3.3,-1.2) [label=left:$$] {};

\draw (node1a) [line width=0.5pt] -- node [left] {} (node2a);
\draw (node1a) [line width=0.5pt] -- node [right] {} (node3a);
\draw (node2a) [line width=0.5pt] -- node [left] {} (node4a);
\draw (node3a) [line width=0.5pt] -- node [right] {} (node5a);
\draw (node4a) [line width=0.5pt] -- node [left] {} (node6a);
\draw[black,dashed] (3,0.5) circle (2);
\node[place, white] (node8) at (3.0,-1.6) [label=below:$\mathcal{G}_2$] {};



\draw (node1) [line width=0.5pt, dashed] -- node [left] {} (node2a);
\draw (node2) [line width=0.5pt, dashed] -- node [left] {} (node1a);

\end{tikzpicture}
\caption{A topologically controllable graph merged by two topologically controllable graphs with two edges.}
\label{network_ex_twopath2}
\end{figure}
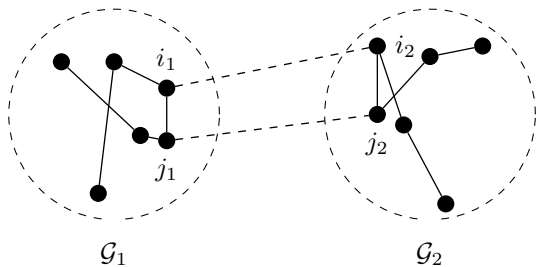

The \textit{Lemma~\ref{lemma_two_quasi}} may provide an intuition for a more general case for merging two graphs. Next, let us consider a case of merging by connecting two edges. 
\begin{lemma} \label{lemma_two_nodes}
Let us consider two disconnected topologically controllable graphs $\mathcal{G}_1$ and  $\mathcal{G}_2$. Let  the state nodes $i_1, j_1$ in $\mathcal{G}_1$ and state nodes $i_2, j_2$ in $\mathcal{G}_2$ be connected by undirected edges as $(i_1, i_2)$ and $(j_1, j_2)$. Then the merged graph $\mathcal{G} = \mathcal{G}_1 \cup \mathcal{G}_2 $ is topologically controllable, if  $\alpha$, $\forall \alpha \subseteq \{i_1, i_2, j_1, j_2\}$, has at least one dedicated node in $\mathcal{N}(\alpha)\setminus \alpha$.
\end{lemma}
\begin{proof}
 Let us choose arbitrary $\alpha \subseteq \mathcal{G}$, where $\alpha = \alpha_1 \cup \alpha_2$, and $\alpha_1 \subseteq \mathcal{G}_1$ and $\alpha_2 \subseteq \mathcal{G}_2$. When we choose $\alpha = \alpha_1$ or $\alpha = \alpha_2$, for any $i \in \alpha$, there is at least one dedicated node $j \in \mathcal{N}(\alpha)\setminus \alpha$ in $\mathcal{G}_1$ or in $\mathcal{G}_2$.

In the case there exist $i$ and $j$ such that $i, j \in \alpha\subseteq \mathcal{G} \setminus \{i_1, i_2, j_1, j_2\}$, and $i \in \alpha_1 \subset \alpha$ and $j \in \alpha_2 \subset \alpha$, there  is still no chance of having $\{\mathcal{N}_i \cap (\mathcal{N}(\alpha)\setminus \alpha)\} \cap  \{     \mathcal{N}_j \cap (\mathcal{N}(\alpha)\setminus \alpha), \forall i, j \}  \neq \emptyset$. Moreover, for all $\alpha_1 \subset \alpha$, and for all $\alpha_2  \subset \alpha$, it is certain that either in $\alpha_1$ or in $\alpha_2$,  there is a node that has a dedicated node in $\mathcal{N}(\alpha_1)\setminus \alpha_1$ or in  $\mathcal{N}(\alpha_2)\setminus \alpha_2$, respectively. Next, let 
$\alpha = \alpha' \cup \alpha''$, where $\alpha' \subseteq \mathcal{G} \setminus\{i_1, i_2, j_1, j_2\}$ and $\alpha'' \subseteq \{i_1, i_2, j_1, j_2\}$, and $\alpha'' \neq \emptyset$. If $\alpha' \neq \emptyset$, it is clear that $\alpha$ has at least one dedicated node. Otherwise, if  $\alpha' = \emptyset$, then it is required that whatever we choose $\alpha''$, where $\alpha'' \subseteq \{i_1, i_2, j_1, j_2\}$, it needs to have at least one dedicated node, which completes the proof.
\end{proof}
Fig.~\ref{network_ex_twopath2} depicts a topologically controllable graph produced by merging two topologically controllable graphs with two edges. Whatever $\alpha \subseteq \{i_1, i_2, j_1, j_2\}$, it has at least one dedicated node. However, in the case of Fig.~\ref{network_ex_twopath3}, when we choose $\alpha=\{j_1, i_2\}$, these nodes have $i_1, j_2$ as the common neighbor nodes. Thus, they do not have any dedicated node. Now, with the above lemmas, by induction, we can make the following theorem.
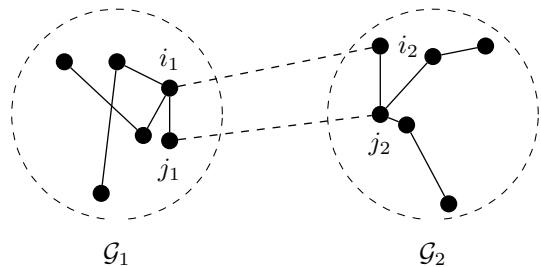
\begin{figure}
\centering
\begin{tikzpicture}[scale=0.7]
\node[place, black] (node1) at (-2,1) [label=above:$i_1$] {};
\node[place, black] (node2) at (-2,0) [label=below:$j_1$] {};
\node[place, black] (node3) at (-3,1.5) [label=below:$$] {};
\node[place, black] (node4) at (-2.5,0.1) [label=below:$$] {};
\node[place, black] (node5) at (-4,1.5) [label=left:$$] {};
\node[place, black] (node6) at (-3.3,-1.0) [label=left:$$] {};

\draw (node1) [line width=0.5pt] -- node [left] {} (node2);
\draw (node1) [line width=0.5pt] -- node [right] {} (node3);
\draw (node1) [line width=0.5pt] -- node [left] {} (node4);
\draw (node3) [line width=0.5pt] -- node [right] {} (node6);
\draw (node4) [line width=0.5pt] -- node [left] {} (node5);
\draw[black,dashed] (-3,0.5) circle (2) ;
\node[place, white] (node7) at (-3.0,-1.6) [label=below:$\mathcal{G}_1$] {};

\node[place, black] (node1a) at (2,0.5) [label=below:$j_2$] {};
\node[place, black] (node2a) at (2,1.8) [label=right:$i_2$] {};
\node[place, black] (node3a) at (3,1.6) [label=below:$$] {};
\node[place, black] (node4a) at (2.5,0.3) [label=below:$$] {};
\node[place, black] (node5a) at (4,1.8) [label=left:$$] {};
\node[place, black] (node6a) at (3.3,-1.2) [label=left:$$] {};

\draw (node1a) [line width=0.5pt] -- node [left] {} (node2a);
\draw (node1a) [line width=0.5pt] -- node [right] {} (node3a);
\draw (node4a) [line width=0.5pt] -- node [left] {} (node1a);
\draw (node3a) [line width=0.5pt] -- node [right] {} (node5a);
\draw (node4a) [line width=0.5pt] -- node [left] {} (node6a);
\draw[black,dashed] (3,0.5) circle (2);
\node[place, white] (node8) at (3.0,-1.6) [label=below:$\mathcal{G}_2$] {};



\draw (node1) [line width=0.5pt, dashed] -- node [left] {} (node2a);
\draw (node2) [line width=0.5pt, dashed] -- node [left] {} (node1a);

\end{tikzpicture}
\caption{Not topologically controllable graph when merged by two topologically controllable graphs with two edges.}
\label{network_ex_twopath3}
\end{figure}
\begin{theorem}
Let two graphs $\mathcal{G}_1$ and $\mathcal{G}_2$ be topologically controllable, respectively. Let $q$ nodes from  $\mathcal{G}_1$ (i.e., let them be denoted as $i_1, i_2, \ldots, i_q$) and another $q$ nodes from  $\mathcal{G}_2$ (i.e., let them be denoted as $j_1, j_2, \ldots, j_q$) be connected one by one. Then, the merged graph $\mathcal{G} = \mathcal{G}_1 \cup \mathcal{G}_2$ is topologically controllable\footnote{Here, achieving the topological controllability is equivalent to satisfying the  condition of \textit{Corollary~\ref{corollary_Tsatsomeros}}. } if and only if  $\alpha$, $\forall \alpha \subseteq \{i_1, \ldots, i_q, j_1, \ldots, j_q\}$, has at least one dedicated node in $\mathcal{N}(\alpha)\setminus \alpha$.
\end{theorem}
\begin{proof}
The \textit{if} condition can be proved by an induction of the proof of \textit{Lemma~\ref{lemma_two_nodes}}. For the \textit{only if} condition, let there exist $\alpha$, $\alpha \subseteq \{i_1, \ldots, i_q, j_1, \ldots, j_q\}$, that does not have a dedicated node. Then, there exists at least one $\alpha \subset \mathcal{G}$, which does not satisfy the condition of \textit{Corollary~\ref{corollary_Tsatsomeros}}.
\end{proof}
The above theorem can be further generalized as:
\begin{corollary} \label{corollary_main2}
Let two graphs $\mathcal{G}_1$ and $\mathcal{G}_2$ be topologically controllable, respectively. Let $q$ nodes from  $\mathcal{G}_1$ (i.e., let them be denoted as $i_1, i_2, \ldots, i_q$) and  $p$ nodes from  $\mathcal{G}_2$ (i.e., let them be denoted as $j_1, j_2, \ldots, j_p$), where $p \neq q$,  be connected. Then, the merged graph $\mathcal{G} = \mathcal{G}_1 \cup \mathcal{G}_2$ is topologically controllable if and only if  $\alpha$, $\forall \alpha \subseteq \{i_1, \ldots, i_q, j_1, \ldots, j_p\}$, has at least one dedicated node in $\mathcal{N}(\alpha)\setminus \alpha$.
\end{corollary}


\section{Topologically Controllability of A Graph}\label{sec_graph}
In the previous section, we have developed conditions for the topologically controllability when merging two graphs. So, starting from a nominaly controllable graph (ex, a path graph), we can enlarge the graph gradually to make a bigger controllable graph. However, the conditions given in the previous section are not applicable for checking the topological controllability of a given network. This section provides algorithms for examining the topological controllability of a graph. That is, given a big size graph $\mathcal{G}$, we would like to examine the topological controllability of the graph. It is not computationally feasible to check all $\alpha \subseteq \mathcal{G}$ whether each $\alpha$ would satisfy the condition of \textit{Corollary~\ref{corollary_Tsatsomeros}}. We propose an algorithm to solve this issue. Let there be $n$ state nodes as $x_1, \ldots, x_n$, and $m$ input nodes as $u_1, \ldots, u_m$. Assume that each input node $u_i$ is solely connected to one state node by one-to-one (injective) mapping. Without loss of generality, let $u_i$ be connected to $x_i$. Then the algorithm starts from the state nodes $x_i$. The key idea is to assign a set of some state nodes to one of input nodes such that the assigned state nodes to a specific input node could be connected by a path, without any cycle. Then, the assigned state nodes with an input node can be considered as a topologically controllable graph (i.e., since it is a path). This process is called \textit{decomposition process}. After that, we would like to examine whether two path graphs can be merged as a topologically controllable graph with the connected edges between two path graphs.
For a notational purpose, the following formal definition is necessary.
\begin{definition}
Consider an undirected path $x_i \leftrightarrow x_i^{1,j_1} \leftrightarrow x_i^{2,j_2} \cdots \leftrightarrow x_i^{k,j_k}$, where $x_i, x_i^{1,j_1}, x_i^{2,j_2}, \cdots, x_i^{k,j_k}$ are nodes connected to the root state node $x_i$ in the graph $\mathcal{G}$.
\begin{itemize}
\item It is the length $k$ path, and denoted as $\mathbf{p}[1:k]$.
\item The node $x_i^{p,j_p}$ is called a descendant node of $x_i^{q,j_q}$ when $p > q$; otherwise if $q > p$, it is called a ancestor node. An immediate descendant node is a child node, and an immediate ancestor node is a parent node.
\item The node $x_i$ is called the starting (root) node and the node $x_i^{k,j_k}$ is the terminal node.
\item When a child node $x_{j}$ is added to $\mathbf{p}[1:k]$, the addition is denoted as $\mathbf{p}[1:k] + x_{j}$ and it becomes a length $k+1$ path $\mathbf{p}[1:k+1]$.
\end{itemize}
\end{definition}
When we seek a path, newly added nodes can be considered as child nodes. But, to be a child node, we need to have a rule. Let $i=1$. Then, starting from $x_i$, we search neighbor nodes of $x_i$, i.e., $\mathcal{N}_{x_i}$, which are children nodes of $x_i$. Then from the nodes $x_i^{1,j_1} \in \mathcal{N}_{x_i}$, we also choose neighbor nodes of
$x_i^{1,j_1}$ as  $x_i^{2,j_2} \in \mathcal{N}_{x_i^{1,j_1}}$. If $x_i^{2,j_2}$ is not connected to $x_i$, then it is considered as a child. Similarly, from a child node $x_i^{k,j_k}$, we also search neighbor nodes as $x_i^{k+1,j_{k+1}}$. If $x_i^{k+1,j_{k+1}}$ is not connected to any of  $\{x_1, \ldots, x_m\} \cup \{ x_i^{1,j_1}, \ldots, x_i^{k-1, j_{k-1}} \}$, then it is considered as a child. By this way, we would find a path for node $i$, which is denoted as $\bar{\mathbf{p}}_{i}$. After obtaining the final path $\bar{\mathbf{p}}_{i}$ for $x_i$, we update $i$ as $i \gets i +1$. When $i \geq 2$, we repeat the above process; but $x_i^{k+1,j_{k+1}}$ should not be connected to any of  $\{x_1, \ldots, x_m\} \cup \{ x_i^{1,j_1}, \ldots, x_i^{k-1, j_{k-1}} \}$ and any nodes in the previously searched paths $\bar{\mathbf{p}}_{1}, \ldots, \bar{\mathbf{p}}_{i-1}$.

\begin{definition} (Children nodes) Suppose that we have obtained the final path  $\bar{\mathbf{p}}_{1}$ with the starting node $x_1$, $\ldots$, the final path $\bar{\mathbf{p}}_{i-1}$ with the starting node $x_{i-1}$. Then, for $x_i$, from a node $j$, search all neighbor nodes. The neighbor nodes, which are not connected directly to (i) ancestor nodes of $j$, (ii) $x_i, i=1, \ldots, m$, and (iii)  any nodes in the previously searched paths $\bar{\mathbf{p}}_{1}, \ldots, \bar{\mathbf{p}}_{i-1}$,  are called children nodes (denoted as $\mathcal{C}_i$).
\end{definition}

\begin{definition} (Path update)
Let a length $k$ path $\mathbf{p}[1:k]$ be given with the terminal node $x_k$.  The terminal node $x_k$ has a set of children nodes $\mathcal{C}_{x_k}$. Then, the path $\mathbf{p}[1:k]$ is updated to a set of length $k+1$ paths as $\mathbf{p}[1:k+1] \in \mathcal{P}[1: k+1] \triangleq \{    \mathbf{p}[1:k] + x_j, ~\forall x_j \in      \mathcal{C}_{x_k}   \}$.
Thus, if the cardinality of the set $\mathcal{C}_{x_k}$ is $\beta$, i.e., $\vert \mathcal{C}_{x_k} \vert = \beta$, then the cardinality of the set $ \mathcal{P}[1: k+1]$ is also $\beta$. The set of paths $\mathcal{P}[1: k+1]$ is called updated path set of the path $\mathbf{p}[1:k]$.
\end{definition}
From the above definition, when a path is given as $\mathbf{p}[1:k]$, the updated path set exists if and only if the terminal node of the path $\mathbf{p}[1:k]$ has children nodes. Given a path set $\mathcal{P}[1: k]$, let the paths $\mathbf{p}[1:k] \in \mathcal{P}[1: k]$ be denoted as $\mathbf{p}_1, \ldots, \mathbf{p}_{f}$, where $f = \vert \mathcal{P}[1: k]\vert$. The terminal node of each path $ \mathbf{p}_{j}, ~j \in \{1, \ldots, f\}$ is denoted as $\bar{x}_{j,k}$. With the above definitions, the path search algorithm can be produced as in \textit{Algorithm~\ref{path_search}}.
\begin{algorithm}
\caption{Path search algorithm (decomposition process)}\label{path_search}
\begin{algorithmic}[1]
\Procedure{}{}
\State Obtain $ x_i,~\forall i =1, \ldots, m$ from $u_i$
\State $i=0$
\BState \emph{path}:
 \For{$i=i+1$} 
   \State $k=0$
      \State Select children nodes of $x_i$
     \State Generate the path set $\mathcal{P}[1: 1]$
    \For{$k=k+1$}

      \State Let $\vert \mathcal{P}[1: k] \vert=f$
    \State $j=1$; $ \mathcal{P}[1: k+1]  =\emptyset$; $\mathcal{C}=\emptyset$


      \While{$j \not= f+1$}         

      \State Select $\mathbf{p}_{j} \in \mathcal{P}[1: k]$
     \State Select the terminal node $\bar{x}_{j,k}$ of $\mathbf{p}_{j}$
       \State Select children nodes of $\bar{x}_{j,k}$ and denote the set of children nodes as $\mathcal{C}_{\bar{x}_{j,k}}$
     \State $\mathcal{C} = \mathcal{C} \cup \mathcal{C}_{\bar{x}_{j,k}}$
      \State Make updated path set $\mathcal{P}^j[1: k+1]$ of $\mathbf{p}_{j}$
     \State $\mathcal{P}[1: k+1] =  \mathcal{P}[1: k+1]  \cup \mathcal{P}^j[1: k+1]$
       \State $j \gets j+1$
      \EndWhile\label{euclidendwhile}

     \If { $\mathcal{C} =\emptyset$}
    \State Select the longest path from the paths in $\mathcal{P}[1: k+1]$ and denote it as $\bar{\mathbf{p}}_{i}$
    \State \textbf{goto} \textit{path}
    \EndIf
  \EndFor   \Comment{End for $k$}

  \EndFor   \Comment{End for $i$}

\State \textbf{output} $\bar{\mathbf{p}}_{i}$ for all $i =1, \ldots, m$
\EndProcedure
\end{algorithmic}
\end{algorithm}

The outputs of \textit{Algorithm~\ref{path_search}} are the paths $\bar{\mathbf{p}}_{i}$ for all $i =1, \ldots, m$. Let these path graphs be denoted as $\mathcal{G}_1=(\mathcal{V}_1, \mathcal{E}_1), \ldots, \mathcal{G}_m=(\mathcal{V}_m, \mathcal{E}_m)$. They can be merged by \textit{Corollary~\ref{corollary_main2}}. We first merge $\mathcal{G}_1$ and $\mathcal{G}_2$. For this, we need to find all the edges connecting two graphs $\mathcal{G}_1$ and $\mathcal{G}_2$. If these edges satisfy the condition of  \textit{Corollary~\ref{corollary_main2}}, then two graphs are merged for a single graph which is also topologically controllable. Otherwise, we need to choose maximum edges that connect two graphs under the topologically controllable condition.  When the graphs $\mathcal{G}_1, \mathcal{G}_2, \ldots, \mathcal{G}_k$ are merged as a single graph, it is written as $\mathcal{G}[1:k]$. The \textit{Algorithm~\ref{graph_merging}} outlines the graph merging process.  To the algorithm, we first make a reverse version of  \textit{Corollary~\ref{corollary_main2}} as:
\begin{definition}\label{coro_largest} (Largest edge merging)
Let two graphs $\mathcal{G}_1$ and $\mathcal{G}_2$ be connected by a set of edges $\mathcal{E}' = \{ (i',j'), i' \in \mathcal{G}_1, ~j' \in \mathcal{G}_2\}$. The largest subset of $\mathcal{E}'$, which makes the merged graph topologically controllable,  is  
\begin{align}
\mathcal{E}'' = \arg_{\mathcal{E}^\ast}\max\{ \vert \mathcal{E}^\ast \vert \}
\end{align}
if $\alpha$, $\forall \alpha \subseteq \{  i^\ast, j^\ast: ~(i^\ast,j^\ast) \in \mathcal{E}^\ast \}$, has at least one dedicated node in $\mathcal{N}(\alpha)\setminus \alpha$.
\end{definition}

\begin{algorithm}
\caption{Graph merging algorithm}\label{graph_merging}
\begin{algorithmic}[1]
\Procedure{}{}
\State $\bar{\mathbf{p}}_{i} \to \mathcal{G}_i,~\forall i =1, \ldots, m$
\State $i=0$
\BState \emph{merge}:
 \For{$i=i+1$} 
      \State Given two graphs $\mathcal{G}[1:i]$ and $\mathcal{G}_{i+1}$, find $\forall (i,j)$ where $i \in \mathcal{G}[1:i]$ and $j \in \mathcal{G}_{i+1}$
     \State Choose the largest edge set  $\mathcal{E}''$ that satisfies \text{Corollary~\ref{coro_largest}}
    \State Merge two graphs  $\mathcal{G}[1:i]$ and $\mathcal{G}_{i+1}$ as a single graph $\mathcal{G}[1:i+1]$
       with the edge set $\mathcal{E}''$
     \If { $i < m$}
    \State \textbf{goto} \textit{merge}
    \Else
    \State
     \textbf{goto} \textit{end}
    \EndIf

  \EndFor   \Comment{End for $i$}

\BState \emph{end}:

\State \textbf{output}  $\mathcal{G}[1:m]$
\EndProcedure
\end{algorithmic}
\end{algorithm}

With  \textit{Algorithm~\ref{graph_merging}}, let us suppose that we have obtained $\mathcal{G}[1:m]=\mathcal{G}^\dag(T^\dag) = (\mathcal{V}^\dag, \mathcal{E}^\dag)$. Then, the graph $\mathcal{G}^\dag$ is topologically controllable, and the nodes $\overline{\mathcal{V}}= \mathcal{V} \setminus \mathcal{V}^\dag$ are not ensured to be topologically controllable. On the other hand, if $\mathcal{V}_1 \cup \mathcal{V}_2 \cup \cdots \cup \mathcal{V}_m = \mathcal{V}$, then the edges  $\overline{\mathcal{E}}= \mathcal{E} \setminus \mathcal{E}^\dag$ are harmful for a  topological controllability of the nominal network $\mathcal{G}$, and need to be removed for topological controllability. In this sense, we can claim the following conclusion:
\begin{theorem} \label{theorem_final_main}
If $ \mathcal{V} = \mathcal{V}^\dag$ and $ \mathcal{E} = \mathcal{E}^\dag$, then the nominal graph $\mathcal{G}$ is topologically controllable.
\end{theorem}

The overall procedure to examine the topological controllability of a network can be summarized as follows. Given a network, it is required to transform the network to a graph $\mathcal{G}(T) =(\mathcal{V}, \mathcal{E})$. Then, by \textit{Algorithm~\ref{path_search}}, for all inputs $u_i$, we search for the paths $\bar{\mathbf{p}}_{i}$.  Then, by  \textit{Algorithm~\ref{graph_merging}}, by way of finding the largest edge set $\mathcal{E}''$, we gradually merge the paths to have a topological controllable graph $\mathcal{G}[1:m]=\mathcal{G}^\dag$.

All the results of this section and previous section were developed under the \textit{Assumption~\ref{assum_L_matrix}}. Thus, it may be necessary to check whether  \textit{Assumption~\ref{assum_L_matrix}} is satisfied or not. For this, we may need to check whether the graph $\mathcal{G}^\dag$ is $L$-matrix or not. For this, from $\mathcal{G}^\dag$, we obtain $T^\dag$ as the inverse of ${\mathcal{G}^\dag}(T^\dag)$. If $T^\dag$  is a $L$-matrix, then the network corresponding to the graph $\mathcal{G}^\dag$ is concluded as topologically controllable. The $L$-matrixness of a matrix $T$ can be examined using some existing results; for example, refer to \cite{Brualdi1995sign}.

\section{Examples} \label{sec_exam}
Let us consider the network depicted in Fig.~\ref{network_graph_concept}(a). To use  \textit{Algorithm~\ref{path_search}} for the path search, the labels of nodes should be changed as $x_6 = u_1$, $x_7 = u_2$, and $x_8= u_3$. Then, the state nodes $x_1, \ldots, x_5$ are searched with new-labels, as per the \textit{Algorithm~\ref{path_search}}. By the algorithm, we can obtain three paths $\mathcal{G}_1 = \bar{\mathbf{p}}_{1} :  u_1 (= x_6)  \rightarrow  x_3 \rightarrow x_2$, $\mathcal{G}_2 =  \bar{\mathbf{p}}_{2} :   u_2 (= x_7)  \rightarrow x_4$, $\mathcal{G}_3 =  \bar{\mathbf{p}}_{3} :  u_3 (= x_8)  \rightarrow  x_5 \rightarrow x_1$. Now, it is required to apply   \textit{Algorithm~\ref{graph_merging}} for merging these path graphs. There are two edges connecting $\mathcal{G}_1$ and $\mathcal{G}_2$, i.e., $(3,4)$ and $(2,4)$. These edges are in $\mathcal{E}''$. Thus, the merged graph $\mathcal{G}[1:2]$ is a topological controllable graph. Then, there are two graphs $\mathcal{G}[1:2]$ and $\mathcal{G}_3$, which needs to be merged. There are four edges between them, i.e., $(1,2), (1,3), (2,5), (5,4)$. It is also easy to check that their edges are also in $\mathcal{E}''$. Consequently, we can conclude that the original network (depicted in Fig.~\ref{network_graph_concept}(a)) or its corresponding graph (depicted in Fig.~\ref{network_graph_concept}(b)) is topologically controllable. This conclusion is confirmed from a number of numerical random tests,  with random values in the elements of $L$, by checking the rank of the following controllability Gramian matrix:
\begin{align}
\mathbf{C}_{L} = [ B, LB, L^2 B, L^3 B, L^4 B] \nonumber
\end{align}
For all random tests, and for any specific cases (with all edge values being $1$ or $-1$), the rank was $5$.

Next, let us consider another network depicted in Fig.~\ref{network_graph_concept_non}. It is a network, which is same to Fig.~\ref{network_graph_concept}(a), but with one more edge $(1,4)$. As same to the case of Fig.~\ref{network_graph_concept}(a), we can have three path graphs $\mathcal{G}_1 = \bar{\mathbf{p}}_{1} :  u_1 (= x_6)  \leftarrow  x_3 \leftrightarrow x_2$, $\mathcal{G}_2 =  \bar{\mathbf{p}}_{2} :   u_2 (= x_7)  \leftarrow x_4$, $\mathcal{G}_3 =  \bar{\mathbf{p}}_{3} :  u_3 (= x_8)  \leftarrow  x_5 \leftrightarrow x_1$. When merging graphs $\mathcal{G}[1:2]$ and $\mathcal{G}_3$, unlikely  Fig.~\ref{network_graph_concept}(a), there are five edges $(1,2), (1,3), (1,4), (2,5), (5,4)$. Then, when we choose $\alpha =\{1,2\}$, the nodes $1$ and $2$ share nodes $3, 4$, and $5$ as the common neighboring nodes; hence in this case, nodes $1$ and $2$ do not have any dedicated node. Thus, the condition for the topological controllability is not satisfied (i.e., as per \textit{Theorem~\ref{theorem_final_main}}, we have $\mathcal{V} = \mathcal{V}^\dag$; but $\mathcal{E} \setminus \mathcal{E}^\dag \neq \emptyset$).
 From a number of random tests, with random values in the elements of $L$, the rank of $\mathbf{C}_{L}$ was still $5$. But, with some specific values, for examples, edge values being $1$ or $-1$, or with integer values, the rank of $\mathbf{C}_{L}$ was not full. For example, Fig.~\ref{example1_random} shows the results of random tests. The left plot shows the rank of $\mathbf{C}_{L}$, when $(1,4) \neq 0$. But, when $(1,4)$ switches to zero, the rank becomes $5$, as shown in the right plots. Consequently, we can see that the edge $(1,4)$ is harmful for the topological controllability.
\begin{figure}
\centering
\begin{tikzpicture}[scale=0.7]
\node[place, black] (node1) at (-1,1) [label=above:$1$] {};
\node[place, black] (node2) at (-2,0) [label=below:$2$] {};
\node[place, black] (node3) at (-1,-1) [label=below:$3$] {};
\node[place, black] (node4) at (0,-1) [label=below:$4$] {};
\node[place, black] (node5) at (0.2,0.8) [label=right:$5$] {};

\node[place, circle] (node6) at (-2.0,-1.2) [label=below:$6$] {};
\node[place, circle] (node7) at (1,-1.5) [label=right:$7$] {};
\node[place, circle] (node8) at (1.0,1.5) [label=right:$8$] {};

\draw (node1) [line width=0.5pt] -- node [left] {} (node2);
\draw (node1) [line width=0.5pt] -- node [right] {} (node3);
\draw (node1) [line width=0.5pt] -- node [right] {} (node4);
\draw (node1) [line width=0.5pt] -- node [right] {} (node5);
\draw (node2) [line width=0.5pt] -- node [right] {} (node3);
\draw (node2) [line width=0.5pt] -- node [left] {} (node4);
\draw (node2) [line width=0.5pt] -- node [left] {} (node5);
\draw (node3) [line width=0.5pt] -- node [left] {} (node4);
\draw (node4) [line width=0.5pt] -- node [left] {} (node5);

\draw (node6) [-latex, line width=0.5pt] -- node [right] {} (node3);
\draw (node7) [-latex, line width=0.5pt] -- node [right] {} (node4);
\draw (node8) [-latex, line width=0.5pt] -- node [right] {} (node5);
\end{tikzpicture}
\caption{A possibly non-topological controllable network.}
\label{network_graph_concept_non}
\end{figure}
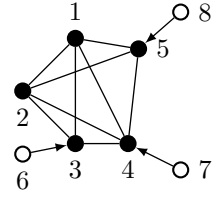
To check the controllability under the same signs, let the signs of edges be given as $\text{sign}(a_{12})= +$,
$\text{sign}(a_{13})= -$,  $\text{sign}(a_{14})= -$, $\text{sign}(a_{15})= -$, $\text{sign}(a_{23})= -$,  $\text{sign}(a_{24})= -$, $\text{sign}(a_{25})= -$, $\text{sign}(a_{34})= -$, and $\text{sign}(a_{45})= +$. Then, we randomly assign absolute magnitude of edges in integer values $1, 2, 3, 4, 5$.  Fig.~\ref{example1_sign} shows the random test results. When  $(1,4) \neq 0$, the rank is reduced; but when it switches to zero, the rank becomes full. But, surprisingly, when the sign of $a_{14}$ changes to $\text{sign}(a_{14})= +$, or the sign of $a_{13}$ changes to $\text{sign}(a_{13})= +$, the rank becomes full again.  Fig.~\ref{example1_sign2} shows the rank tests with different signs. Thus, from numerical random tests, we can see that the topological controllability is dependent on the sign of edges.
\begin{figure*}[h]
\centering
\includegraphics[width=14cm,height=5.2cm]{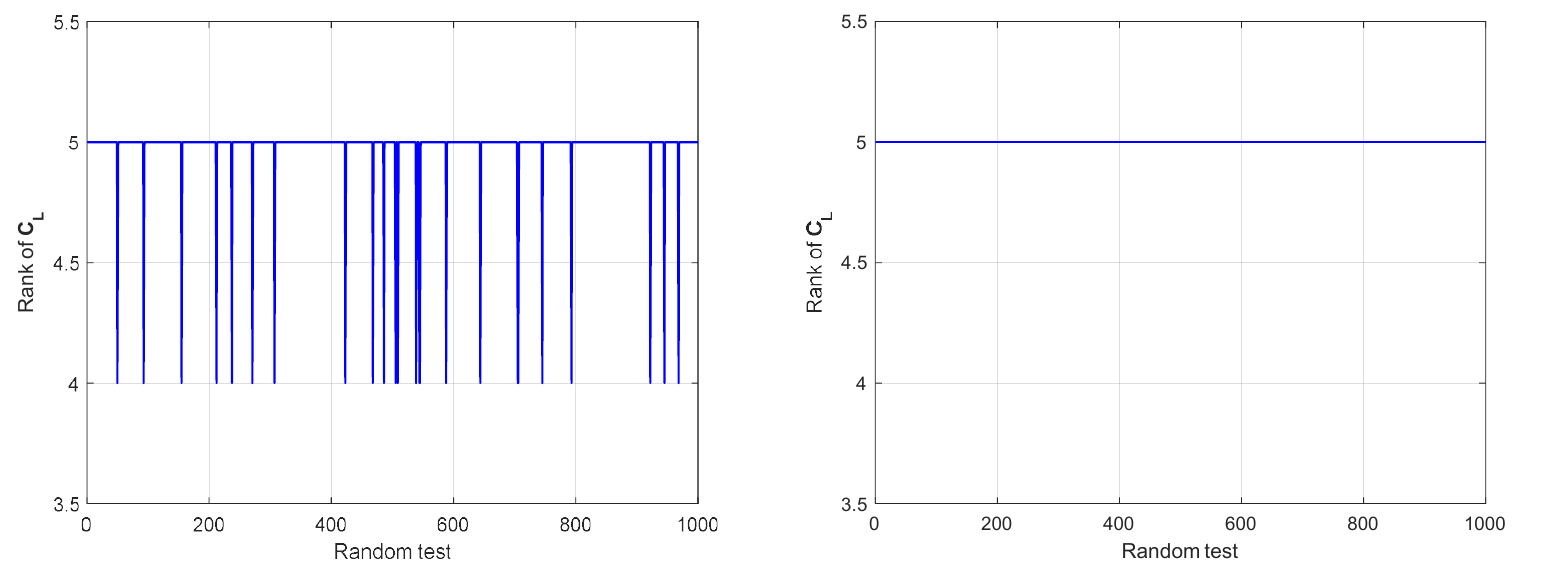}
\caption{Random tests for chekcing the rank of $\mathbf{C}_L$. Left: $(1,4) \neq 0$. Right: $(1,4)=0$} \label{example1_random}
\end{figure*}
\begin{figure*}[h]
\centering
\includegraphics[width=14cm,height=5.2cm]{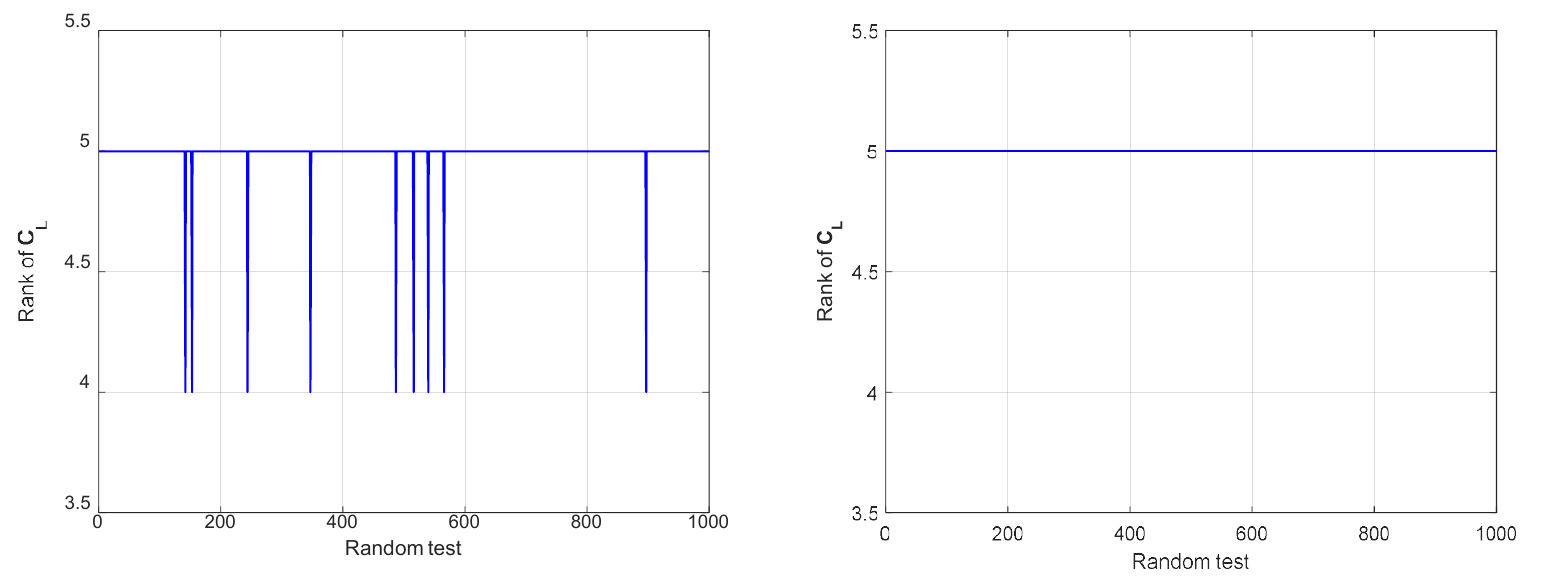}
\caption{Random tests for chekcing the rank of $\mathbf{C}_L$ under the same signs. Left: $(1,4) \neq 0$. Right: $(1,4)=0$} \label{example1_sign}
\end{figure*}
\begin{figure*}[h]
\centering
\includegraphics[width=14cm,height=5.2cm]{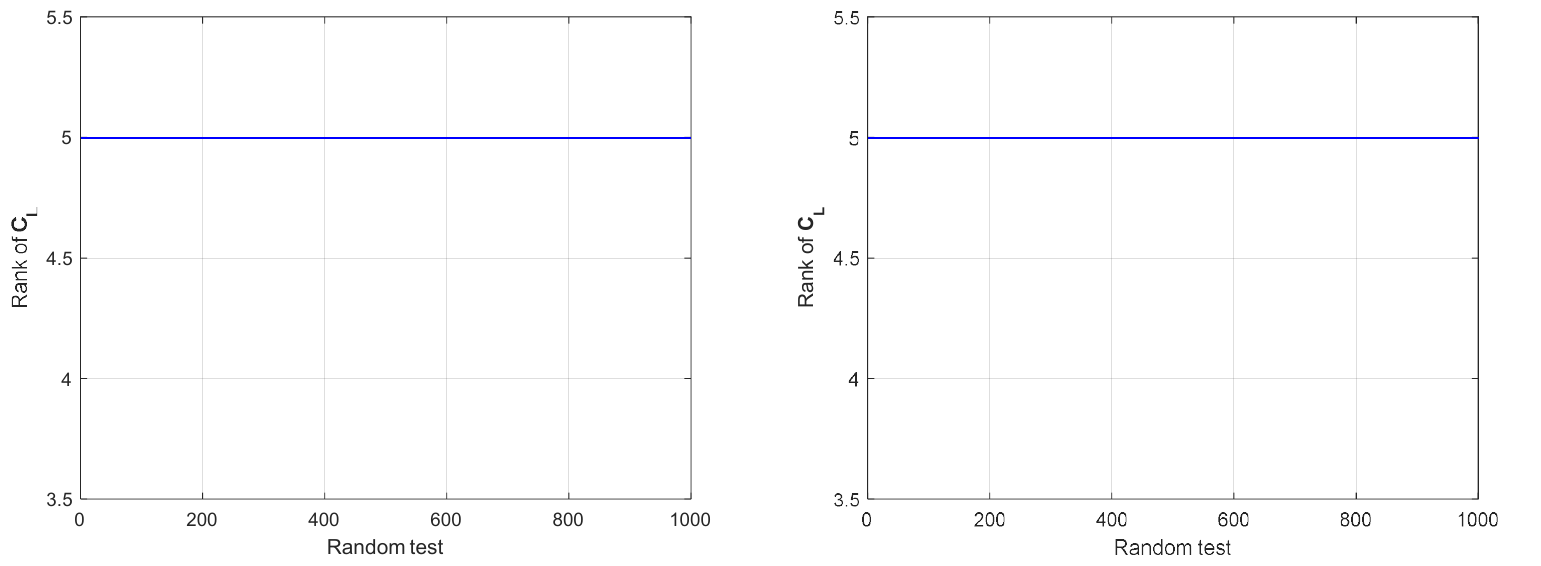}
\caption{Random tests for chekcing the rank of $\mathbf{C}_L$ under the same signs. Left: $(1,4) \neq 0$; but  $\text{sign}(a_{14})= +$. Right: $(1,4)\neq 0$,  $\text{sign}(a_{14})= -$ and $\text{sign}(a_{13})= +$.} \label{example1_sign2}
\end{figure*}


\section{Conclusion} \label{sec_conc}
This paper has presented conditions to establish the controllability of an undirected networks of diffusively-coupled agents using only the knowledge of the signs of edges, motivated by and based on results in \cite{Tsatsomeros_siam_1998}. Because the resulting conditions are computationally-hard, we developed a merging process for creating an enlarged network starting from a basic controllable graph. The merging process was then used to develop a decomposition process for evaluating the topological controllability of a given network. Through numerical simulations, we could verify the effectiveness of the proposed algorithms. However, there are still many open problems. For example, if we could find basic path graphs in the decomposition process in an optimal way (i.e., minimizing the number of nodes that are not included in the final paths), then we may be able to find a more bigger subgraph induced by the controllability\footnote{It appears that the process for finding the paths in an optimal way looks similar to the maximum matching process proposed in \cite{YangYu_nature_2011}. However, it seems that the merging and decomposition algorithms proposed here are more efficient and general. Also, we do not claim that the path graphs are only basic controllable subgraphs, although our algorithms were developed from path graphs. Thus, in our future efforts, we would like to develop new decomposition and merging algorithms from more general basic controllable subgraphs.}. In this paper, we have focused on undirected diffusive-coupled networks, but we believe we can easily extend the results to the directed case. These extensions will be studied in our future research.

\section*{Acknowledgment}
The work of this paper has been supported by the National Research Foundation (NRF) of Korea under the grant NRF-2017R1A2B3007034. 


\bibliographystyle{IEEEtran}
\bibliography{opinion_dynamics}

\begin{thebibliography}{10}
\providecommand{\url}[1]{#1}
\csname url@samestyle\endcsname
\providecommand{\newblock}{\relax}
\providecommand{\bibinfo}[2]{#2}
\providecommand{\BIBentrySTDinterwordspacing}{\spaceskip=0pt\relax}
\providecommand{\BIBentryALTinterwordstretchfactor}{4}
\providecommand{\BIBentryALTinterwordspacing}{\spaceskip=\fontdimen2\font plus
\BIBentryALTinterwordstretchfactor\fontdimen3\font minus
  \fontdimen4\font\relax}
\providecommand{\BIBforeignlanguage}[2]{{%
\expandafter\ifx\csname l@#1\endcsname\relax
\typeout{** WARNING: IEEEtran.bst: No hyphenation pattern has been}%
\typeout{** loaded for the language `#1'. Using the pattern for}%
\typeout{** the default language instead.}%
\else
\language=\csname l@#1\endcsname
\fi
#2}}
\providecommand{\BIBdecl}{\relax}
\BIBdecl

\bibitem{ChingTaiLin_TAC_1974}
C.-T. Lin, ``Structural controllability,'' \emph{IEEE Trans. Automatic
  Control}, vol.~19, no.~3, pp. 201--208, 1974.

\bibitem{Antsaklis:2007:LSP:1543534}
P.~J. Antsaklis and A.~N. Michel, \emph{A Linear Systems Primer}.\hskip 1em
  plus 0.5em minus 0.4em\relax Birkh\"{a}user Basel, 2007.

\bibitem{Goldstein_TCS_1979}
L.~Goldstein, ``Controllability$/$observability analysis of digital circuits,''
  \emph{IEEE Transactions on Circuits and Systems}, vol.~26, no.~9, pp.
  685--693, 1979.

\bibitem{Feng_Lu_ICMLC_2005}
X.-Y. Feng and K.-S. Lu, ``Structural controllability and reducibility of {RLC}
  networks with bipolar transistor,'' in \emph{Proc. of the 2005 International
  Conference on Machine Learning and Cybernetics}, 2005, pp. 1015--1020.

\bibitem{Summers_etal_TCNS_2015}
T.~H. Summers, F.~L. Cortesi, and J.~Lygeros, ``On submodularity and
  controllability in complex dynamical networks,'' \emph{IEEE Transactions on
  Control of Network Systems}, vol.~3, no.~1, pp. 91--101, 2015.

\bibitem{Wang_automatica_2016}
L.~Wang, G.~Chen, X.~Wang, and W.~K.~S. Tang, ``Controllability of networked
  mimo systems,'' \emph{Automatica}, vol.~69, no.~7, pp. 405--409, 2016.

\bibitem{Wang_etal_SR_2017}
L.-Z. Wang, Y.-Z. Chen, W.-X. Wang, and Y.-C. Lai, ``Physical controllability
  of complex networks,'' \emph{Scientific Reports}, vol.~7, pp.~--, Jan. 2017.

\bibitem{Shen_etal_CCDC_2018}
C.~Shen, Z.~Ji, and H.~Yu, ``The structural controllability of edge dynamics in
  complex networks,'' in \emph{Proc. of the 2018 Chinese Control and Decision
  Conference (CCDC)}, 2018, pp. 5356--5360.

\bibitem{Gu_etal_nature_comm_2015}
S.~Gu, F.~Pasqualetti, M.~Cieslak, Q.~K. Telesford, A.~B. Yu, A.~E. Kahn, J.~D.
  Medaglia, J.~M. Vettel, M.~B. Miller, S.~T. Grafton, and D.~S. Bassett,
  ``Controllability of structural brain networks,'' \emph{Nature
  Communications}, vol.~6, pp.~--, Oct. 2015.

\bibitem{YangYu_nature_2011}
Y.-Y. Liu, J.-J. Slotine, and A.-L. Barab\'{a}si, ``Controllability of complex
  networks,'' \emph{Nature}, vol. 473, pp. 167--173, May, 2011.

\bibitem{Ruths_science_2014}
J.~Ruths and D.~Ruths, ``Control profiles of complex networks,''
  \emph{Science}, vol. 343, pp. 1373--1375, March, 2014.

\bibitem{Mesbahi_Egerstedt_book}
M.~Mesbahi and M.~Egerstedt, \emph{Graph Theoretic Methods in Multiagent
  Networks}.\hskip 1em plus 0.5em minus 0.4em\relax Princeton University Press,
  2010.

\bibitem{Hou_TCS1_2016}
B.~Hou, X.~Li, and G.~Chen, ``Structural controllability of temporally
  switching networks,'' \emph{IEEE Transactions on Circuits and Systems I:
  Regular Papers}, vol.~63, no.~10, pp. 1771--1781, 2016.

\bibitem{Partovi_etal_ICRAM_2010}
A.~Partovi, H.~Lin, and Z.~Ji, ``Structural controllability of high order
  dynamic multi-agent systems,'' in \emph{Proc. of the 2010 IEEE Conference on
  Robotics, Automation and Mechatronics}, 2010, pp. 327--332.

\bibitem{Faradonbeh_TCNS_2017}
M.~K.~S. Faradonbeh, A.~Tewari, and G.~Michailidis, ``Optimality of
  fast-matching algorithms for random networks with applications to structural
  controllability,'' \emph{IEEE Transactions on Control of Network Systems},
  vol.~4, no.~4, pp. 770--780, 2017.

\bibitem{Yang_etal_ICCA_2007}
J.~xiong Yang, D.~yun Lin, and M.~qing Li, ``Analysis on controllability of
  descriptor systems under structural decomposition,'' in \emph{Proc. of the
  IEEE International Conference on Control and Automation}, 2007, pp.
  3169--3172.

\bibitem{Tsatsomeros_siam_1998}
M.~Tsatsomeros, ``Sign controllability: {S}ign patterns that require complete
  controllability,'' \emph{SIAM Journal on Matrix Analysis and Applications},
  vol.~19, pp. 355--364, 1998.

\bibitem{Brualdi1995sign}
R.~A. Brualdi and B.~L. Shader, \emph{Matrices of Sign-Solvable Linear
  Systems}.\hskip 1em plus 0.5em minus 0.4em\relax Cambridge University Press,
  1995.

\end{thebibliography}

\end{document}